\documentclass[conference,9pt]{IEEEtran}
\IEEEoverridecommandlockouts

\usepackage[ruled, vlined, linesnumbered]{algorithm2e}
\SetKw{Continue}{continue}
\SetKw{Break}{break}
\SetKw{From}{from}
\SetKw{To}{to}
\SetKw{True}{True}
\SetKw{False}{False}

\usepackage{amsmath,bm}
\usepackage{amssymb}
\usepackage{graphicx}
\usepackage{subfig}
\usepackage{xcolor}
\usepackage{multirow}
\usepackage{cases}
\usepackage{cite}
\usepackage{amsthm}
\usepackage{url}

\newtheorem{theorem}{Theorem}

\newtheorem{remark}{Remark}
\newtheorem{problem}{Problem}

\newtheorem{definition}{Definition}
\newtheorem{proposition}{Proposition}

\SetCommentSty{mycommfont}

\allowdisplaybreaks

\newenvironment{myitemize}{\begin{list}{$\bullet$}
{\setlength{\topsep}{1mm}
\setlength{\itemsep}{0.25mm}
\setlength{\parsep}{0.25mm}
\setlength{\itemindent}{0mm}
\setlength{\partopsep}{0mm}
\setlength{\labelwidth}{15mm}
\setlength{\leftmargin}{4mm}}}{\end{list}}

\begin{document}
\title{%Online Control Adaptation for Energy Saving \\ of Human-CPS with Safety Guarantee \\
Opportunistic Intermittent Control with Safety Guarantees \\for Autonomous Systems 
\thanks{This work is supported by the National Science Foundation awards 1834701, 1834324, 1839511, 1724341, and 1646497. It is also funded in part by the DARPA BRASS program under agreement number FA8750-16-C-0043 and ONR grant N00014-19-1-2496.}}

\author{
    \IEEEauthorblockN{Chao Huang\IEEEauthorrefmark{1}, Shichao Xu\IEEEauthorrefmark{1}, Zhilu Wang\IEEEauthorrefmark{1}, Shuyue Lan\IEEEauthorrefmark{1}, Wenchao Li\IEEEauthorrefmark{2}, Qi Zhu\IEEEauthorrefmark{1}}
    \IEEEauthorblockA{\IEEEauthorrefmark{1}Northwestern University, 
    \\ \{chao.huang, qzhu\}@northwestern.edu, \{ShiChaoXu2023, ZhiluWang2018, ShuyueLan2018\}@u.northwestern.edu}
    \IEEEauthorblockA{\IEEEauthorrefmark{2}Boston University, wenchao@bu.edu }
}

\maketitle

\begin{abstract}
Control schemes for autonomous systems are often designed in a way that anticipates the worst case in any situation.
At runtime, however, there could exist opportunities to leverage the characteristics of specific environment and operation context for more efficient control. 
In this work, we develop an online intermittent-control framework that combines formal verification with model-based optimization and deep reinforcement learning to opportunistically \emph{skip} certain control computation and actuation to save actuation energy and computational resources without compromising system safety. 
Experiments on an adaptive cruise control system demonstrate that our approach can achieve significant energy and computation savings.

\end{abstract}


\section{Introduction}

For safety-critical autonomous systems such as robots and automated vehicles, control schemes are often designed conservatively so that system safety can be maintained in a wide variety of situations~\cite{chisci2001systems,richards2005robust,lofberg2003minimax}. %\zhu{Better to use robust control citations here.}. 
During the operation of these systems, however, such schemes can be overly conservative and result in unnecessary resource and/or energy consumption. 
This paper first makes the observation that certain control steps, even if they are skipped, do not impact either the performance or safety of the overall system. Armed with this observation, 
we propose an online scheme that opportunistically skips control computation and the corresponding actuation steps 
by learning specific characteristics of the system's operating environment. 
We further show that safety could be maintained with this more efficient control scheme.
%In this paper, we propose to identify control steps that By learning and leveraging the characteristics of specific operation environment and context, more efficient control methods may be used. 

Consider the example of an adaptive cruise control (ACC) system, in which an ego vehicle automatically adjusts its speed to  maintain a safe distance from the vehicle in front. To ensure the safety across a variety of situations (e.g. an aggressive front vehicle vs. a conservative front vehicle), the ego vehicle may adopt a safe control scheme such as one that based on robust model predictive control (RMPC) \cite{chisci2001systems}. 
At each control step, the RMPC calculates the actuation signal based on the two vehicles' speeds and their relative distance. However, it does not make any prediction on the front vehicle's intent. 
In practical scenarios, the front vehicle may exhibit certain behavior patterns, e.g., an aggressive driver that accelerates and decelerates frequently, or stop-and-go in a traffic jam. 
We argue that we can design more computation/energy-efficient control schemes by learning and exploiting these patterns. 
The key is how to learn these patterns quickly and how to guarantee system safety when certain control steps are skipped. 
%Controller synthesis for a safety-critical Cyber-physical system (CPS) is a classical problem and have been adequately studied in past few years \cite{zhao2012hybrid,larsen2015safe,fan2018controller,fremont2018reactive,hussien2018lazy}. To ensure the safety, a controller needs to fulfill the safety requirements under any possible uncertainties. However, such controllers may be over conservative when the uncertainties are brought by humans, where the human involved CPS is also known as Human-CPS \cite{li2014synthesis,gorecky2014human,seshia2015formal,sadigh2019verifying}. The reason behind is that different from classical uncertainties, such as wind or friction, uncertainties brought by human may have distinct behavior pattern. Due to the lack of understanding the human behaviors, the controller designed for all possible uncertainties can still ensure the safety, but suffer low performance in some cases, for instance energy cost.

%\li{add citation on simplex architectures that include a safety controller} 
In this work, we consider systems with an existing safe controller, and develop a novel online intermittent-control framework to opportunistically skip the computation and actuation of the underlying controller by leveraging the characteristics of specific operation context and environment. 
For instance, at each control step of the above ACC example, our method will decide whether to run the underlying RMPC and apply its actuation/control input, or simply apply a zero control input. Such opportunistic skipping could help save computational resources for running the underlying control algorithm and save actuation energy by applying zero control inputs. 
To \emph{achieve both safety and efficiency, our method addresses two key challenges}: 1) How to ensure the system safety when zero input is applied at some control steps? 2) How to effectively leverage the characteristics of specific operation context and environment? 

For the first challenge, our framework uses formal analysis to guarantee the system safety under skipping of controls. 
%In this paper, we focus on leveraging the behavior pattern to improve the energy cost performance, where such a pattern can be either known or unknown by the system. Specifically, we formulate it as an online decision problem, which we call it \emph{control adaptation}: for a given controller, at each time step $t$, decide whether to use the given controller to compute a control input, or directly apply zero input, denoted as \emph{control adaptation variable}. By using the knowledge of the behavior pattern, the overall energy consumption can be expected reduced. 
%The difficulty to solve this decision problem lies on two aspects. First, how to ensure the system safety when zero inputs are applied in some time steps. To handle this problem, 
Specifically, we first compute a \emph{strengthened safe set} based on the notion of \emph{robust control invariant} and \emph{backward reachable set} of the underlying safe controller. Intuitively, the strengthened safety set represents the states at which the system can accept any control input at the current step and be able to stay within safe states, with the underlying safe controller applying input from the next step on.  
We then develop a monitor to check whether the system is within such strengthened safe set at each control step. Whenever it is found that the system state is out of the strengthened safe set, the monitor will require the system to apply the underlying safe controller for guaranteeing system safety.
%Once the system state found out of the strengthened safe set, the monitor will force the system to apply the given controller, which guarantees the safety.

For the second challenge, we develop two approaches to leverage the characteristics of operation context and environment when the system is within the strengthened safe set, depending on the type of the underlying safe controller and whether the characteristics are known explicitly. In the simpler case where the safe controller has an analytic expression and the characteristics can be explicitly captured, %(e.g., in the ACC example, the speed profile of the front vehicle is known), 
we use a model-based approach to decide the skipping choices by solving a mixed integer programming (MIP) program. Otherwise, we use
a deep reinforcement learning (DRL) approach to learn the mapping from the current state and the historical characteristics to the skipping choices, which implicitly reflects the impact of specific operation context and environment.

%Second, when the system is within the strengthened safe set, how to leverage the behavior pattern of the perturbance to decide the value of the control adaptation variable. We propose two approaches depending on the type of the controller and whether the pattern is known explicitly. In the simple case where the controller is an analytic expression and perturbance is known, we use a model-predictive-control (MPC) like method to predict the control adaptation variable by solving a mixed integer programming (MIP). Otherwise, we use deep reinforcement learning to learning the mapping between the current state together with the perturbance history and the control adaptation variable. Such mapping implicitly include the impact of the underlying pattern of the perturbance.

\smallskip
\noindent
\textbf{Related work:} Our work is related to the rich literature on weakly-hard systems and fault-tolerant control systems. In weakly-hard systems, occasional deadline misses are allowed for control computation in a bounded manner, e.g., the typical $(m,K)$ constraint allows at most $m$ deadline misses in any $K$ consecutive control instances%which is represented as $(m,K)$ constraint -- at most $m$ failures in $K$ consecutive tasks
~\cite{Hamdaoui_IEEE95,gujarati2019iteration,huang2019exploring}. In fault-tolerant control systems, broader fault types (e.g., sensing, actuation, or system errors) and fault models (e.g., stochastic model) are considered~\cite{jiang-arc12}. However, while these works try to preserve the system safety~\cite{Frehse_RTSS_14,duggirala2015analyzing,huang2019formal} or stability~\cite{lan2016new,song2016adaptive} under \emph{passive} faults, our work considers skipping control operations \emph{proactively} to save resources for control computation and reduce energy for actuation.
%we actively generate ``faults'' in the controller so that it skips unnecessary control input actuation.

Our work is also related to methods on safe reinforcement learning~\cite{junges2016safety,fulton2018safe,alshiekh2018safe}, as we also restrict the possible system actions to a safe set for ensuring safety. The difference is that our approach computes the safe action set by deriving the robust control invariant and backward reachable set of an underlying controller and then leverages the safe set for proactively exploring control skippings.  
%The commonality is that these works, as well as ours, restrict the possible actions of the agent to a safe action set to ensure the system safety, while the difference is our work consider a specific CPS control problem such that the safe action set can be easily estimated by the notion of robust control invariant and backward reachable set.

\smallskip
\noindent
In summary, this work makes the following novel contributions.
\begin{myitemize}
    \item We develop a novel online intermittent-control framework for opportunistically skipping the computation and actuation of an underlying safe controller to save computational resources and actuation energy. Our framework can be generally applied to various underlying controllers in a discrete linear time-invariant system, and achieves both safety and efficiency.
    % actively applying zero inputs by leveraging the knowledge of disturbance, to improve the energy cost performance.
    \item Our framework ensures safety by developing a formal method to compute a strengthened safe set of the underlying controller. 
   % We ensure the safety of the adaptation framework by computing a strengthened safe state set.
    \item Our framework achieves efficiency by developing a model-based approach and a DRL-based approach to decide the skipping choices of the underlying controller under different scenarios. The model-based approach is applied when the underlying controller has an analytic expression and the characteristics of specific operation context and environment are explicitly known; while the DRL approach is applied otherwise.
   % We propose a model-based approach, and reinforcement learning approach for deciding the adaptation in the cases of different types of control law and disturbances respectively. 
    \item We demonstrate the effectiveness of our approach through extensive experiments on an ACC case study, using the widely-used SUMO (Simulation of Urban Mobility) simulator~\cite{SUMO2012}.
    %show the performance of the proposed approach by an exhaustive study on adaptive cruise control example with various setting.
\end{myitemize}

%The rest of the paper is organized as follows. Section~\ref{sec:problem} introduces the overall problem formulation of our opportunistic intermittent control. Section~\ref{sec:approach} presents the general design of our intermittent-control framework, and uses RMPC as an example of the underlying safe controller to illustrate our ideas. Section~\ref{sec:experiments} shows the experimental results on an ACC case study, and Section~\ref{sec:conclusion} concludes the paper.

%In the experiments, we perform an exhaustive study on the example of adaptive cruise control. In this example, the velocity of the front car is considered as a bounded disturbance depending on the human drive's behavior pattern. We consider the different setting of the pattern to see the performance comparison between our adaptation approach and simply applying the given controller all the time. The experiments show that in most cases, our approach can achieve a better performance on energy cost.

\section{Problem Formulation}
\label{sec:problem}

%\subsection{Control Adaption Problem}
We consider a discrete linear time-invariant (LTI) system described by the following difference equation:
\begin{equation} \label{eq:dynamic}
	x(t+1) = Ax(t)+Bu(t)+w(t), \qquad t\geq 0,
\end{equation}
where $x\in \mathcal{R}^n$ is the state variable, $u\in \mathcal{R}^m$ is the control input variable, and $w$ is a bounded perturbation. $A$ and $B$ are transformation matrix. We assume that the state space of the system is $\mathcal{R}^n$. The constraints on the safe state, the control input, and the bounded perturbation are 
\begin{equation} \label{cons:stateinputConstraint}
    x(t) \in X, \qquad u(t) \in U, \qquad w(t) \in W,
\end{equation}
where $X\subset \mathcal{R}^n$, $U\subset \mathcal{R}^m$ and $W\subset \mathcal{R}^k$ are polytopes, and $0 \in X$, $0\in U$, $0\in W$. Note that $X$ represents the set of safe states, and perturbation $w(t)$ captures the characteristics of the specific operation context and environment. 
%Specifically, in this paper, we focus on the perturbance caused by \emph{humans}, rather than classical disturbance, such as wind or friction.

We assume that such system can be controlled by a safe feedback controller $\kappa$. At each sampling instant $t=0,1,\cdots $, the system reads the state $x(t)$ of the plant and feeds it to the controller $\kappa$ to obtain the control input $u=\kappa(x)$. The new input $u(t)$ will be applied at the next time step. We use 1-norm of the input $\|u(t)\|_1$ to represent the energy cost at each control step. 

A well-designed controller $\kappa$ may ensure the system safety by sustaining any possible perturbation. However, when applying in practice, such design may be over-conservative and energy-consuming. 
Thus, the question is whether we can save actuation energy (and possibly also computational resources) by \emph{skipping} the computation and actuation of control inputs (i.e. using zero inputs) at some steps. This requires effectively leveraging the dynamic characteristics of specific operation context and environment at runtime (reflected via the pattern of perturbation $w(t)$), while guaranteeing system safety.  
%in some periods via leveraging the underlying behavior pattern of the perturbance, while maintaining the safety performance. 

We use a binary indicator $z(t)$ at time step $t$ to represent the skipping choice: $1$ denotes actuating the control input computed by the controller $\kappa$, and $0$ denotes skipping the control and applying zero input. Our online opportunistic intermittent-control problem can be formulated as follows. 
%Intuitively, the behaviour pattern can be various and thus instead of an off-line adaptation approach, we consider the following online control adaptation problem:
%
\begin{problem}[Online Opportunistic Intermittent-Control Problem] \label{pr:adaptation}

    Given a dynamical system defined in Equation~\eqref{eq:dynamic} and a controller $\kappa$ that can ensure system safety, i.e., $x(t)\in X$, the opportunistic intermittent-control problem is to determine the skipping choice variable $z(t)$ at each control instant $t$, such that the system will stay within $X$ and the overall energy cost $\sum_{i=0}^\infty \|u(t)\|$ is minimized.

    %Given a dynamical system defined in Equation~\eqref{eq:dynamic} and a controller $\kappa$ with a robust control invariant $X_I$, at each control instant $t$, %let $z(t)$ be a binary control variable for each time step $t$, where 1 denotes applying the control input computed by $\kappa$ and 0 denotes applying zero input respectively. 
   % the opportunistic intermittent-control problem is to determine $z(t)$, such that from any initial state $x(0) \in X_I$, the system will stay within $X_I$ and the overall energy cost $\sum_{i=0}^\infty \|u(t)\|$ is minimized.
\end{problem} %\zhu{Redefined the problem based on the general notion of safe set. At this point we have not even introduced robust control invariant yet.}

\section{Opportunistic Intermittent-Control Framework} 
\label{sec:approach}

%In this paper, we propose a general framework to handle the aforementioned difficulties and generate $z(t)$ at each time step $t$ while ensuring safety. The main idea is to relax the infinite-time optimization into finite one. In terms of different types of control algorithms and disturbances, we either explicitly encode the finite-time optimization in a model predictive manner, or use a deep reinforcement learning algorithm to implicitly learn the relation between the current state and the optimal control adaptation. It is worthy noting that such relaxation may lead to a violation of the safety specification. Thus, in order to maintain the safety, we check \emph{strengthened} safety constraints that are derived based on reachability analysis at each time step. Once the strengthened safety constraints are satisfied, $z(t)$ can be freely chosen by our adaptation approach. Otherwise, $z(t)$ is forced to be 1 such that the given controller can drive the system to a more 'safe' state.

There are two key aspects of our approach: 1) ensuring that the system always stays within the safe state space $X$, and 2) deciding the skipping choice variable $z(t)$ by leveraging the characteristics of perturbation $w$. 

For the first aspect, to ensure system safety, we define three safe state sets of different levels, namely the \emph{original safe set} $X$, the \emph{robust invariant set} $X_I$, and the \emph{strengthened safe set} $X'$, as shown in Figure~\ref{fig:safeLevel}. The original safe set $X$ is the largest of the three, and given by the problem definition. While the system is still safe within this set, there is no guarantee that it will stay within $X$ for the next time step, even with the input from the underlying safe controller $\kappa$. Thus, we consider the robust invariant set $X_I \subseteq X$. When the system is within this set, it is still controllable and can remain in this set by applying the controller $\kappa$. Finally, for considering skipping controls, we define the strengthened safe set $X' \subseteq X_I$. When the system is within this set, it will stay within $X_I$ (and thus controllable and safe) for the next time step, regardless of the skipping choice at the current step (i.e., regardless of whether $z(t)$ is $1$ or $0$). Intuitively, the goal of our framework is to make skipping choices when the system is within the strengthened safe set $X'$, and to apply the underlying safe controller $\kappa$ whenever the system goes out of $X'$ but is still within the robust invariant set $X_I$ ($\kappa$ will be applied until the system goes back to $X'$). Note that the system is guaranteed to never go out of $X_I$ and thus maintain safety.
\begin{figure}[tbp]
	\centering
	\includegraphics[width=0.35\textwidth]{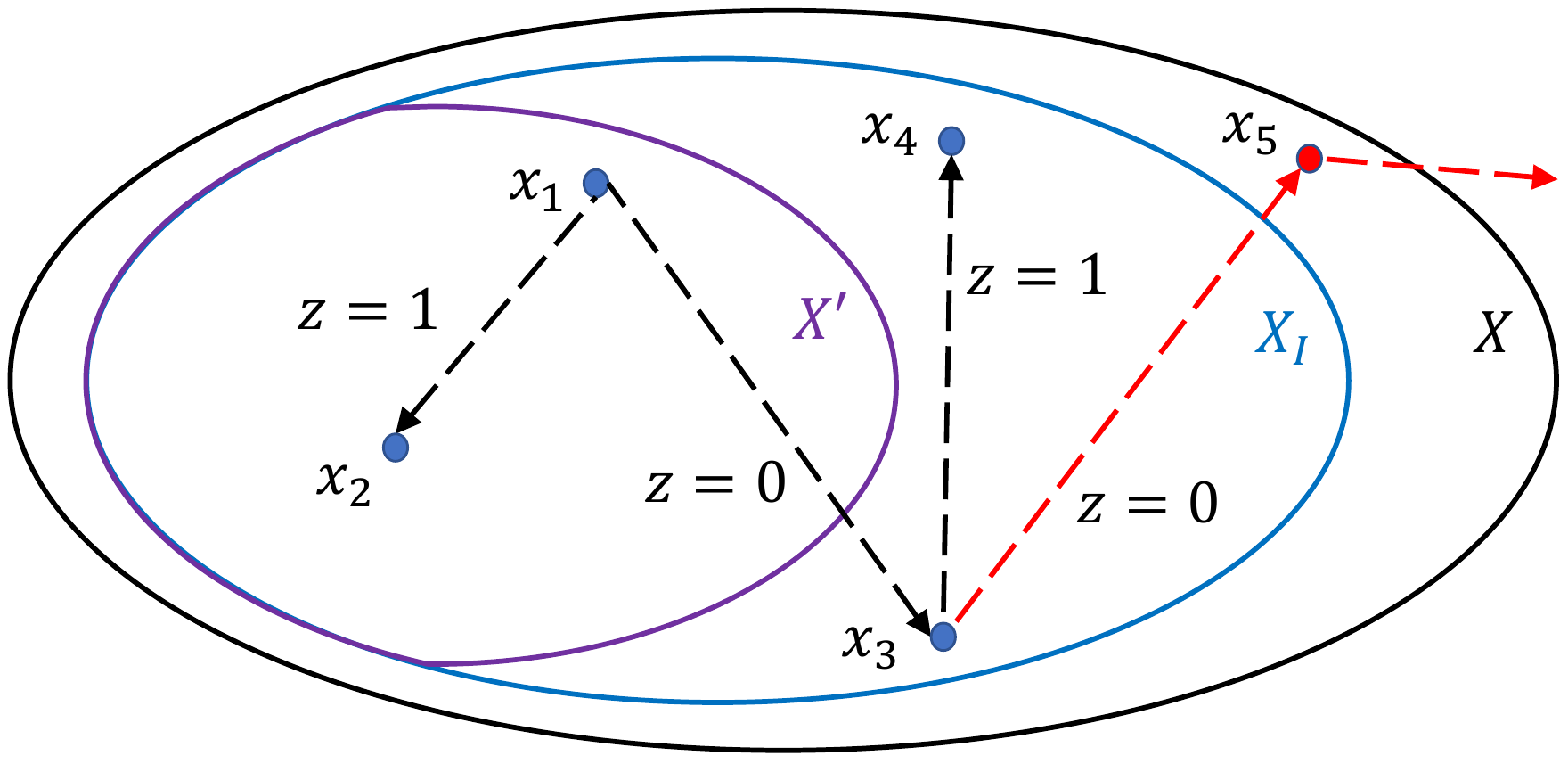}
	\caption{Three safe state sets of different levels: \emph{strengthened safe set} $X'$, \emph{robust invariant set} $X_I$, and \emph{original safe set} $X$. For instance, for $x_1 \in X'$, in the next time step the system may steer to either $x_2$ or $x_3$ that are both within $X_I$ and hence safe and controllable. For $x_3 \in X_I - X'$, the system may steer to $x_4$ in the next step by applying actuation from $\kappa$, which is still controllable. However, if applying zero input at $x_3$, the system may go out of $X_I$ and steer to $x_5 \in X - X_I$. $x_5 $ is a safe state for now. However there is no guarantee that the system will remain within $X$ in the next step, even with the actuation from $\kappa$.}
	\label{fig:safeLevel}
\end{figure}

%The key idea of this paper is to \textit{leverage the behavior pattern of the human-involved perturbance\li{perturbance?} to improve the energy performance of a conservatively designed controller by actively choosing proper adaptation variable $z(t)$ at each time step $t$}. Thus, a deep reinforcement learning (DRL) based approach is proposed for such task. 

%Meanwhile, to ensure safety, we define three safe sets of different levels, namely \emph{robust invariant set} $X_I$, \emph{strengthened safe set} $X'$ and \emph{potentially unsafe set} $X$ (Figure \ref{safeLevel}). Intuitively, strengthened safe set ensures the highest safety level: the system will be safe and controllable at next time step regardless of the value of adaptation variable $z$. Robust invariant set is the supper set of the former with the lower safety level: the system is still controllable and will remain controllable if applying the controller $\kappa$. Finally, we re-name the safe set $X$ as potentially unsafe set, which is still safe at present but with no safety guarantee for the next time step. Based on that, an additional monitor is used to check against a \emph{strengthened safe set} $X'$ that are derived based on reachability analysis at each time step. Once the strengthened safety constraints are satisfied, $z(t)$ can be freely chosen by our adaptation approach $\Omega$. Otherwise, $z(t)$ is forced to be 1 such that the given controller to avoid the system to enter $X-X_I$. Such a procedure will continue until the system goes back to $X'$. 
For the second aspect of our approach, to decide the skipping choices, we develop a model-based approach and a DRL-based approach for different scenarios based on whether the underlying controller has an analytic form and whether the characteristics of the perturbation $w(t)$ is known.

\begin{figure*}[tbp]
	\centering
	\includegraphics[width=0.75\textwidth]{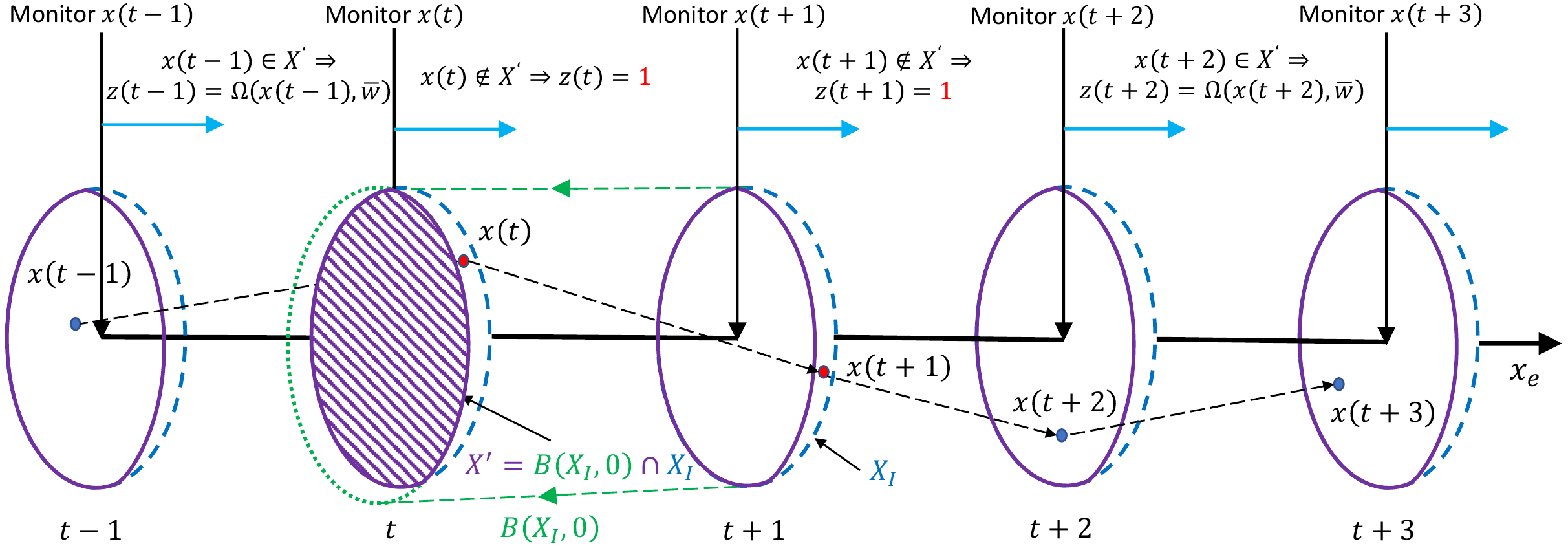}
	\caption{Schematic of our online opportunistic intermittent-control framework. Note that at time step $t-1$ and $t+2$, the system state is in the strengthened safe set $X'$, and thus the skipping choice variable can be freely chosen.}
	\label{fig:strengthensafety}
\end{figure*}

The schematic of our opportunistic intermittent-control framework is shown in Figure~\ref{fig:strengthensafety}. Its flow is shown in Algorithm~\ref{algo:main}. As aforementioned, to ensure safety, we will first compute the robust invariant set $X_I$ and the strengthened safe set $X'$ (lines 1). In our approach, the initial state has to be within $X_I$ (line 2). During operation, at each time step $t$, we will monitor the current state $x(t)$ by collecting sensor inputs. If $x(t)$ is within $X'$, we know that the system will remain within $X_I$ regardless of whether we skip the actuation from $\kappa$ at this step, and we will determine the skipping choice $z(t)$ via a function $\Omega$ that uses the model-based approach or the DRL-based approach (line 6). If $x(t)$ is out of $X'$ (i.e., within $X_I - X'$), we cannot skip the actuation of $\kappa$ and have to set the skipping choice $z(t)$ to $1$ (line 9). Finally, depending on the value of $z(t)$, either the actuation input from $\kappa$ or zero input will be applied (lines 10 to 15).
%algorithm is stated as follows (Figure \ref{fig:strengthensafety}). As aforementioned, to ensure the system safety, we restrict the initial state set as $X_I$ (line 1). At each time step $t$, the monitor will sense the current state (line 3) and the intermittent control action is determined depending on which region the current state stays in: the action is completely determined by the function $\Omega$ when $x(t) \in X'$ (line 5), while $\kappa$ must be applied when $x(t) \in X_I - X'$ (line 7).
%\smallskip
%
\begin{algorithm}[htbp]
    \SetAlgoLined
    %\KwData{Dynamic system (\ref{dynamic}) with state constraint (\ref{stateConstraint}), the control law (\ref{control}) with exponentially stable parameter $(\alpha,\lambda)$, radius of safe state region $d$, $(m,K)$ constraint, sampling period $T$, typical worst-case response time $W$}
    %\KwResult{Safe initial state set $\Theta$}
    Compute robust invariant set $X_I$ and strengthened safe set $X'$\;
    Initialization: $t \gets 0$, $x(0) \in X_I$\;
    \While{true}{
        Monitor the current state $x(t)$ via sensor inputs\;
        \eIf{$x(t) \in X'$}  {
            $z(t) \gets \Omega(x(t),\bar{w}(t))$\;
        }{
            $z(t) \gets 1$\;
        }
        \eIf{$z(t) = 1$}{
            $u(t) \gets \kappa (x(t))$\;
        }{
            $u(t) \gets 0$\;
        }
        Actuate the control input $u(t)$\;
        $t \gets t + 1$\;
    }
\caption{Opportunistic Intermittent-Control Framework}
\label{algo:main}
\end{algorithm}

%\noindent
%\textbf{[Control adaptation algorithm]}
% \begin{enumerate}
%     \item[Step 1:] Initialization: $t \leftarrow 0$, $x(0) \in X'$;
%     \item[Step 2:] Monitor the current state $x(t)$. If $x(t) \in X'$, $z(t) \leftarrow \Omega(x(t),\bar{w}(t))$, where $\bar{w}(t)$ represents a bounded sequence of the history perturbances. Otherwise, $z(t) \leftarrow 1$;
%     \item[Step 3:] If $z(t) = 0$, $u(t) \leftarrow 0$. Otherwise, $u(t) \leftarrow \kappa(x(t))$;
%     \item[Step 4:] Apply $u(t)$ to the system \eqref{eq:dynamic};
%     \item[Step 4:] $t \leftarrow t + 1$. Go to step 2.
% \end{enumerate}

The key components in our framework is the computation of $X_I$ and $X'$ for ensuring safety, and the design of the function $\Omega$ for making skipping decisions. Next, we will introduce their details.

%The key issue of the above algorithm is the computation of the robust invariant set and $X_I$ strengthened safe set $X'$ and the adaptation approach $\Omega$. In the rest of this section, we introduce the details of $X'$ and $\Omega$ respectively.

\subsection{Ensuring Safety: Computation of Robust Invariant Set $X_I$ and Strengthened Safe Set $X'$}
\label{subsec:safety}
%Intuitively, it is not enough to ensure safety by only checking if the state $x(t) \in X$ at each time step $t$ since $z(t)$: the system is already unsafe when you observe it. Thus we need to derive a more restrict safety constraint $X'$ for monitoring. 

To compute $X_I$ and $X'$, we first introduce the concept of robust control invariant set~\cite{rungger2017computing} and backward reachable set~\cite{asarin2000approximate}.
\begin{definition}[Robust Control Invariant Set]
    Given a discrete dynamical system as defined in Equation~\eqref{eq:dynamic} and a controller $\kappa$, the robust control invariant set of the system is defined as:
    \begin{equation} \label{eq:robustInv}
        X_I = \{x\ |\ \forall w \in W, x \in X_I, \exists x' \in X_I, f(x,\kappa(x),w) = x'\}.
    \end{equation}
    \label{def:robustInv}
\end{definition}
\begin{definition} [Backward Reachable Set]
    Given a discrete dynamical system as defined in Equation~\eqref{eq:dynamic}, a controller $\kappa$ and a set $Y$, the (one-step) robust backward reachable set of the system from $Y$ under a skipping choice variable $z$ is defined as:
    \begin{displaymath}
        B(Y,z) = 
        \begin{cases}
            \{x\ |\ \forall w \in W, x' \in Y, f(x,\kappa(x),w) = x'\}, & z=1, \\
            \{x\ |\ \forall w \in W, x' \in Y, f(x,0,w) = x'\}, & z=0.
        \end{cases}
    \end{displaymath}
    \label{def:backwardReach}
\end{definition}
\begin{definition} [Strengthened Safety Set] Based on Definitions~\ref{def:robustInv} and~\ref{def:backwardReach}, we define the strengthened safe set $X'$ as:
\begin{equation} \label{eq:strengthened}
    X' = B(X_I,0) \cap X_I.
\end{equation}
\end{definition}

\begin{theorem}[Safety of Our Approach] \label{th:soundness}
    Let $X'$ and $X_I$ be defined as Equation \eqref{eq:strengthened} and Equation \eqref{eq:robustInv}, respectively, our framework can ensure the system safety for any skipping decision function $\Omega$.
\end{theorem}
\begin{proof}
    Since $X' = B(X_I,0) \cap X_I \subset X$, the system is safe if $x \in X'$. Assume that at time step $t$, the monitor observes that the system goes out of the region $X'$, i.e. $x(t) \in X_I - X'$ (see Figure~\ref{fig:strengthensafety}). While at the last time step $t-1$, the system was still in $X'$, i.e. $x(t-1) \in X' = B(X_I,0) \cap X_I = B(X_I,0) \cap B(X_I,1)$. By the definition of backward reachable set, we know that $x(t) \in X_I$. Thus, based on the definition of robust controlled invariant, the system is controllable and will remain in $X_I$ under $\kappa$. This shows that our framework (Algorithm~\ref{algo:main}) will always maintain system safety. %Thus Step 2 of the algorithm ensure the system safety.
\end{proof}

\smallskip
\noindent
\textbf{Computing $X_I$ and $X'$.} %Computing $X'$ relies on both the robust control invariant $X_I$ and the backward reachable set $B(X_I,0)$. 
We first consider the computation of the robust control invariant set $X_I$. For linear feedback control, namely $\kappa(x) = Kx$, the robust invariant can be computed as \cite{rakovic2005invariant}:
\begin{displaymath}
    X_I = \alpha (W \oplus (A+BK)W \oplus \cdots \oplus (A+BK)^n W),
\end{displaymath}
where $\alpha$ and $n$ ares hyper-parameters. $\oplus$ denotes the Minkowski sum.

For more advanced control algorithms, we take robust model predictive control (RMPC) as an example, which considers the nominal system model and tightened constraints~\cite{chisci2001systems,richards2005robust}. Given a system state $x(t)$, RMPC solves the following optimization:
\begin{equation} \label{eq:RMPC}
	\begin{aligned}
		& \quad J(x(t))\triangleq
		\min_{\bar{u}_N}\sum_{k=0}^{N-1}P\|x(k|t)\|_1 + Q\|u(k|t)\|_1\\
		\text{s.t.\ } & 
		\begin{cases}
			x(k+1|t) = Ax(k|t)+Bu(k|t), \quad 0 {\leq} k{\leq} N{-}1, \\
			x(k|t) \in X(k), \quad 0 {\leq} k{\leq} N, \\
			u(k|t) \in U, \quad 0 {\leq} k{\leq} N{-}1, \\
			x(0|t) = x(t), \quad x(N|t) \in X_t,
		\end{cases}
	\end{aligned}
\end{equation}
where $u(k|t)$ and $x(k|t)$ denote the input and state of $t+k$ predicted at $t$ respectively. $\bar{u}_H = u(0|t), \cdots , u(N-1|t)$, $P$ and $Q$ are the weights of the state cost and energy cost, respectively. $X_t$ is the terminate set to ensure stability~\cite{mayne2000constrained}, and the tightened constraints $X(k)$ is defined recursively as follows.
\begin{displaymath}
    \begin{gathered}
        X(0) = X, \\
        X(k) = \{x\ |\ x\in X(k{-}1)\ \wedge \ x \oplus A^{k-1} W \subseteq X(k{-}1)\}, \ \ k\geq 1.
    \end{gathered}
\end{displaymath}
Let $\bar{u}^*_N$ be the optimal solution. Then the first sample of $\bar{u}^*_N$ will be applied for actuation: $\kappa(x(t)) = u^*(0|t)$.
We leverage the idea of explicit MPC~\cite{tondel2003algorithm} to first analyze the feasible region of MPC, namely, the state set where MPC optimization is feasible. Assuming the feasible set of the given RMPC is $X_F$, we have:
\begin{proposition}
    The feasible set $X_F$ of the RMPC is also its robust control invariant set (i.e., $X_I = X_F$), if the terminal set $X_t$ ensures stability, that is, there exists a robust local controller $\kappa_L$ such that $\kappa_L(x) + w \in X_T$, $\forall x \in X_T, w\in W$. %\cite{mayne2000constrained}.
\end{proposition}
\begin{proof}
    Given any $x(t) \in X_F$, by the definition of $X_F$, we know that the optimization problem of $J(x(t))$ is feasible and let the optimal solution be $\bar{u}^*_N$. We also let the corresponding optimal value of $x(1|t), \cdots , x(N|t)$ be $x^*(1|t), \cdots , x^*(N|t)$. By the methodology of MPC, we know that the system will steer to $x(t+1)=x^*(1|t)+w(t) \in X(0)$. Now we consider the feasibility of the optimization problem $J(x(t+1))$. We construct a control sequence by combining part of the optimal value of $\bar{u}^*_N$ and the local controller $\kappa_L$, namely $u^*(1|t), \cdots , u^*(N-1|t), \kappa_L(x^*(N|t))$. We have
    % \begin{displaymath}
    % \begin{aligned}
    %     & x'(1|t{+}1) \triangleq A(x^*(1|t){+}w){+}Bu^*(1|t) \in X(2){\oplus} AW = X(1), \\
    %     & x'(2|t{+}1) \triangleq A(x^*(2|t){+}Aw){+}Bu^*(1|t) \in X(3){\oplus} A^2W = X(2),\\
    %     & \cdots
    % \end{aligned}
    % \end{displaymath}
    This control sequence is feasible for $J(x(t+1))$. Therefore we have $x^*(1|t) {\in} X_F$, which means $X_F$ is the robust control invariant set.
\end{proof}

After obtaining the robust control invariant $X_I$ for the RMPC, we can compute the backward reachable set $B(X_I,0)$ using the following formula if $A$ is invertible:
\begin{displaymath}
    B(X_I,0) = A^{-1}(X_I \ominus W),
\end{displaymath}
where $\ominus$ denotes the Minkowski difference.

%\begin{remark}
 %   Observing the result of Theorem \ref{th:soundness}, arbitrary strategy $\Omega$ is valid for safety under the control adaptation algorithm. However, as we may see later in the case study of adaptive cruise control, an elaborate adapatation algorithm will help further improve the cost performance. 
%\end{remark}

\subsection{Achieving Efficiency: Design of Skipping Decision Function $\Omega$}
\label{subsec:efficiency}
%\smallskip
%\subsubsection{Simple control law \& known perturbance}
%\noindent
%\textbf{[Simple control law \& known perturbance]} 

%\smallskip
%\noindent
%\textbf{[Alternative solution for simple case]} 

As aforementioned, we develop a model-based approach and a DRL-based approach for making the skipping decisions (function $\Omega$ in Algorithm~\ref{algo:main}), with the system safety guaranteed by Theorem~\ref{th:soundness}.

\subsubsection{Model-based Approach}

If the control law can be represented as an analytic expression and the perturbation is known, i.e. $w(t)$ is known for any $t\geq 0$, we can develop a model-based approach for optimizing the skipping choice made in function $\Omega$. Specifically, at each time step $t$, 
%we can also try a MPC-like approach to derive the control variable $z$. At each time step $t$, 
we solve the following finite-time optimization problem:
\begin{equation} \label{eq:finite}
	\begin{aligned}
		%J(q(t))\triangleq
		& \qquad \qquad \min_{\bar{z}_H,\bar{u}_H,\bar{x}_H}\sum_{k=0}^{H-1}\|u(k|t)\|_1\\
		\text{s.t.\ } & 
		\begin{cases}
			x(k+1|t) {=} f(x(k|t),u(k|t),w(k|t)), \ 0 {\leq} k{\leq} H{-}1, \\
			x(k+1|t) \in X', \ \ u(k|t) \in U, \quad 0 {\leq} k{\leq} H{-}1, \\
			u(k) \in
			\begin{cases}
			    \kappa(x(k)), & z(t)=1, \\
			    0, & z(t)=0.
			\end{cases}
			\quad 0 {\leq} k{\leq} H{-}1, \\
			x(0|t) = x(t),
		\end{cases}
	\end{aligned}
\end{equation}
where $\bar{z}_H = z(0|t), \cdots , z(H-1|N)$, $\bar{u}_H = u(0|t), \cdots , u(H-1|N)$, and $\bar{x}_H = x(0|t), \cdots , x(H|N)$. Let $\bar{z}^*_H$ be the optimal solution. Then the first sample of $\bar{z}^*_H$ will be applied as the skipping decision, i.e., $\Omega(x(t)) = z^*(0|t)$.

\begin{remark}
    The above optimization formulation~\eqref{eq:finite} is in fact similar to MPC \cite{huang2016hierarchical}, as it is also based on the idea of deciding present action by predicting long-term behavior. However, note that our optimization does not encode terminal constraint as in~\cite{mayne2000constrained} as stability does not need to be considered.
    %The commonality is that both our approach and MPC share the idea of computing present action by predicting long-term behavior, while the main difference is we do not encode terminal constraint \cite{mayne2000constrained} since stability is not necessarily considered.
\end{remark}

% \begin{remark}
%     Readers may find Equation \eqref{eq:finite} is quite similar to the model predictive control (MPC) formulation. The commonality is that both our approach and MPC share the idea of computing present action by predicting long-term behavior, while the main difference is we do not encode terminal constraint \cite{mayne2000constrained} since stability is not necessarily considered.
% \end{remark}

%\subsubsection{Complex cases}
%\smallskip
%\noindent
%\textbf{[Complex cases]}

\subsubsection{DRL-based Approach}

Most advanced control schemes such as the MPC cannot be simply represented as an analytic expression.  
%It is worthy noting that most advanced control cannot be simply represented as an analytic expression, for instance MPC. %It indicates that the optimization-based approach is not applicable for most advanced control algorithms, for instance MPC, since MPC itself cannot be represented as explicit expression. 
Furthermore, it is often impossible to know the perturbation $w(t)$ that reflects the specific operation context and environment a prior. In such case, we develop a machine learning approach for the skipping decision function $\Omega$ to learn and leverage the underlying perturbation pattern. Since there is typically a lack of labelled data for such systems in practice, we use a deep reinforcement learning (DRL) based approach rather than supervise learning. Specifically, we design the following DRL agent for $\Omega$.
%one can hardly know the underlying pattern of a human-related perturbance in practice. It indicates that finding a global optimal solution for energy saving is generally impossible. Thus with the safety guarantee by Theorem \ref{th:soundness}, a proper choice is to use machine learning approaches to generate policy $\Omega$. Due to the lack of training data (labelled input/output pairs), we use deep reinforcement learning (DRL), rather than classical supervised learning to generate $\Omega$. First we show the model of DRL considered in this paper:

\smallskip
\noindent
\textbf{Actions.}
The DRL agent generates two types of actions, $0$ or $1$, representing the cases of $z(t) = 0$ or $z(t) = 1$ in skipping decision.

\smallskip
\noindent
\textbf{State.} 
The action of $\Omega$ depends on not only the observation of current system state, but also the information of past perturbations. Thus, we define a hyper-parameter $r$, which represents the memory length of the perturbations. The state for the DRL agent is the set as $s(t) = \{x(t), w(t-r+1), \cdots , w(t)\}$.
%In this case, we consider a 3 dimensional inputs, denoted as $s = \{v1, v2, d\}$, which implies the currently current car speed, current front car speed, the distance between this car and the front car. 

\smallskip
\noindent
\textbf{Reward (Penalty) function.} %(hyper-parameters are determined after finishing the experiments.)
Typically, if the system goes out of the strengthened safe set $X'$ frequently, more non-zero inputs computed by the underlying controller $\kappa$ need to be applied and there is less energy saving. Thus, the goal of the DRL agent should include both minimizing the energy cost $\sum{\|u(t)\|_1}$ directly and maintaining the system within the strengthened safe set $X'$. Following this idea, we design the reward function $R(s_1,z,s_2)$ of DRL with respect to the agent predecessor state $s_1$, the action $z$ and the successor state $x_2$, with consideration of both objectives:
\begin{displaymath}
    \begin{gathered}
        R(s_1,z,s_2) = - w_1 \cdot R_1 - w_2 \cdot R_2, \\
        R_1 = 
        \begin{cases}
             0, & x_2 \in X',\\
             1, & x_2 \in {{X_I}-{X'}},
        \end{cases} \ \
        R_2 = \begin{cases}
             0, & z {=} 0 \wedge x_1 {\in} X', \\
             \|\kappa(x_1)\|_1, & others.
        \end{cases}
    \end{gathered}
\end{displaymath}
where $R_1$ and $R_2$ are the reward for maintaining the system state $x$ in $X'$ and the reward for the current energy cost, respectively. 
%$d(x,X'^c) = \min_{y\in X'^c}\|x-y\|_1$ represent the distance between the current system state $x$ and the uncontrollable region $X'^C$, 
$w_1$ and $w_2$ are the weights for $R_1$ and $R_2$.

%where $r_1, r_2$ is the reward for maintaining the distance and the reward for fuel cost. $w_1, w_2$ is the weight factor for these two rewards, default is $0.8$ and $4.0$. And $d_{gap}, d_{upper}, d_{lower}$ means the current car distance, the upper bound for car distance, the lower bound for car distance. $u$ is the fuel cost for current step.

\smallskip
\noindent
\textbf{Learning process.}
Our DRL agent interacts with the underlying controller $\kappa$ at each time step. For convenience, we let $w(-r+1), \cdots , w(-1)$ be $0$. The DRL agent starts with the initial state $s(0) = \{x(0), w(-r+1), \cdots , w(0)\}$, and generates the first action $z(0)$.  
%The initial state will be feed into the DRL algorithm and the first action $z(0)$ is generated. 
If $z(0) = 1$ and $x_1 {\in} X'$, the underlying controller $\kappa$ will be applied to compute the control input $u(0) = \kappa(x(0))$; otherwise, we let $u(0) = 0$. Once $u(0)$ is applied, the next system state $x(1)$ is generated. We can then 
%and we can obtain the next agent state $x(1)$ by applying $\kappa(x(0))$. Otherwise, we let $u(0) = 0$ and apply it to get the next agent state $x(1)$. We 
calculate the reward and run the policy network based on the reward, and then obtain the action $z(1)$ for the next step.  
%then we get the action $z(1)$ for the next step. 
We repeat this until we reach the maximum steps, which is a hyper-parameter set in advance.

Note that when the DRL agent drives the system state out of the strengthened safe set $X'$ (i.e., into $X_I - X'$), we will apply the underlying controller $\kappa$ to ensure system safety. The DRL agent will also receive a large penalty in such case (as defined above in the reward function), so it can be motivated to keep the system state within $X'$ for more efficient control and better learning of the perturbation pattern.

\section{Experimental Results on an ACC Case Study} \label{sec:experiments}

\begin{figure}[tbp]
	\centering
	\includegraphics[width=0.43\textwidth]{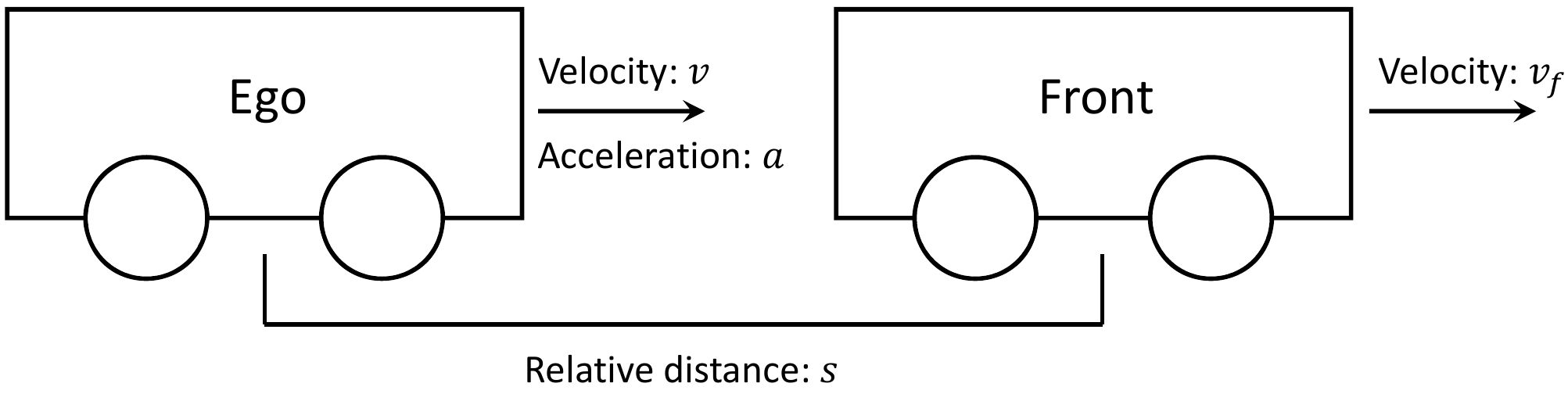}
	\caption{Schematic view of a cruise control scenario.}
	\label{fig:ACC}
\end{figure}

In our experiments, we conduct an extensive case study on an adaptive cruise control (ACC) system~\cite{duggirala2015analyzing,tiwari2003approximate}.
%We consider the adaptive cruise control (ACC) example considered in~\cite{duggirala2015analyzing,tiwari2003approximate}. 
As shown in Figure~\ref{fig:ACC}, there are two vehicles \emph{Ego} and \emph{Front} driving on the road. The \emph{Front} vehicle is moving at a velocity $v_f$ (which may change over time). We are able to control the \emph{Ego} vehicle by tuning its acceleration with a velocity-related resistance. Let $s$ be the relative distance between \emph{Ego} and \emph{Front}, and $v$ be the velocity of \emph{Ego}. The system dynamics follow the standard Newton's laws of motion:
\begin{displaymath}%\label{eq:accDynamic}
    \left\{
    	\begin{array}{lcl}
    		s(t+1) & = & s(t)-(v(t)-v_f(t))\delta, \\
    		v(t+1) & = & v(t)-(kv(t) - u(t))\delta,
    	\end{array}
	\right.
\end{displaymath}
where $(s,v)$ forms the state variable and $u$ is the control input variable. $\delta = 0.1$ is the sampling/control period, and $k=0.2$ is the drag coefficient. 
The velocity of the \emph{Front} vehicle is within the range $v_f \in [30,50]$. 
The ACC system tries to maintain the distance between the two vehicles within a safe range $s \in [120,180]$ for any possible $v_f$.
The \emph{Ego} vehicle has constraints on its velocity and actuation/control input: $v \in [25,55], u \in [-40,40]$.
%In this example, $v_f$ acts as the operation environment and its value may greatly influence the perfo
The control input $u$ is computed by the aforementioned RMPC $\kappa_R$~\cite{chisci2001systems} with the prediction horizon set to $10$. 

As mentioned in Section~\ref{sec:approach}, for an advanced controller as RMPC, we use double deep Q learning \cite{van2016deep} to design the skipping decision function $\Omega$. The hyper-parameters used in DRL are set as follows. The perturbation memory length $r=1$. The weights in the reward function $w_1=0.01$ and $w_2=0.0001$. The training and testing of all the experiments are performed on a desktop with 4-core 3.60 GHz Intel Core i7 and NVIDIA GeForce GTX TITAN. The system is simulated in the SUMO simulator~\cite{SUMO2012}. We evaluate the fuel consumption of $100$ time steps.

\subsection{Overall Effectiveness of Our Approach}

%\subsection{DRL-based Intermittent control, RMPC and bang-bang control: fuel consumption comparison}

We compare our opportunistic intermittent-control approach against the traditional approach of only using the underlying RMPC controller, and against an intuitive bang-bang control scheme based on our framework. The bang-bang scheme uses the same computation of $X_I$ and $X'$ in Section~\ref{subsec:safety} to ensure system safety, but uses a simple strategy (instead of DRL) for deciding skipping choices. 
%In this section, we want to explore the performance of the proposed intermittent control versus RMPC on fuel consumption. Furthermore, we also implement a bang-bang control under our framework, that is,  the bang-bang control will 
Specifically, it applies zero control input whenever the system state is within the strengthened safe set $X'$, and applies the input from RMPC once it is not in $X_I-X'$ :
\begin{equation}
    u(x) = 
    \begin{cases}
         0 & x \in X',\\
         \kappa_{R}(x) & x \in X-X'.\\
    \end{cases}
\end{equation}
%Bang-bang control helps to further characterize the effect of the framework
%where $u_{rmpc}(t)$ is the control input derived from the RMPC. 

In this experiment, we assume that the front vehicle is driving under a sinusoidal velocity variation pattern with a minor disturbance. Specifically, 
\begin{equation}\label{equ:sin_distr}
    v_f(t) = v_e + a_f sin(\frac{\pi}{2}\delta t)+w,
\end{equation}
where $v_e=40$, $a_f=9$, and the random disturbance $w\in[-1, 1]$. We conducted experiments on 500 cases with randomly generated initial state $x(0)$. 

We first compare the fuel consumption of the three approaches (DRL-based opportunistic intermittent-control, bang-bang control, and RMPC only). The fuel consumption data are from SUMO simulations, and directly reflect the actuation energy cost as defined in our Problem~\ref{pr:adaptation}. Figure~\ref{fig:acc_example_3} shows the fuel consumption savings of our DRL-based opportunistic intermittent-control approach and the bang-bang control approach over the traditional method of only using RMPC. The x-axis is the range of fuel consumption saving, and the y-axis shows how many cases (out of the 500) falls into each range for the two approaches. We can clearly see that 1) both the DRL-based opportunistic intermittent-control and the bang-bang control achieve significant savings, showing the potential of skipping controls when possible; 2) the DRL-based opportunistic intermittent-control achieves substantially more savings than the bang-bang control, showing the effectiveness of our DRL-based approach in learning the perturbation pattern and intelligently deciding the skipping choices. Overall, compared with RMPC, the average fuel consumption of bang-bang control is reduced by $16.28\%$, while \textbf{the average fuel consumption of our DRL-based opportunistic intermittent-control is reduced by $23.83\%$}. 

%The fuel consumption improvement statistics is shown in Figure~\ref{fig:acc_example_3}. Overall, comparing with RMPC, the average fuel consumption by bang-bang control at each time step is reduced by $16.28\%$, while the DRL based intermittent control achieves $23.83\%$
%
\begin{figure}[tbp]
	\centering
	\includegraphics[width=0.44\textwidth]{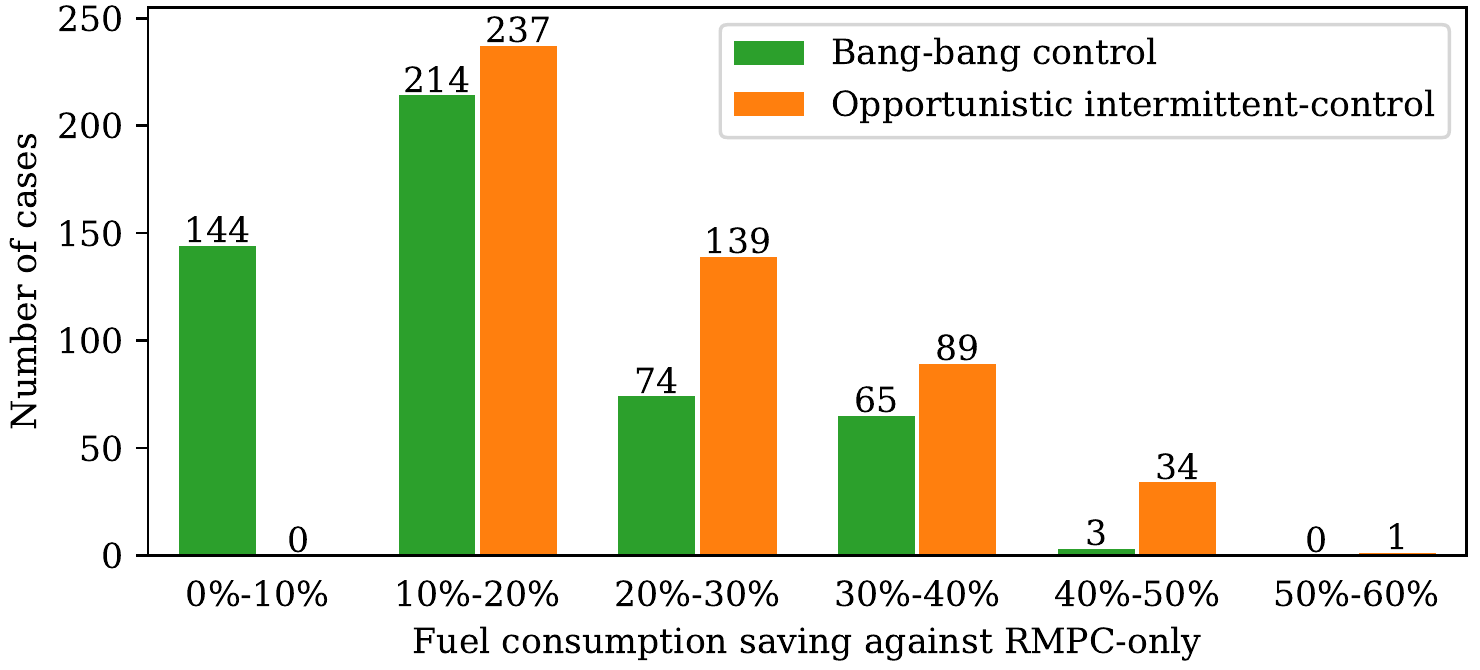}
	\caption{Fuel consumption comparison on 500 test cases. X-axis shows the range of fuel consumption savings of our DRL-based opportunistic intermittent-control and bang-bang control over the traditional method of only using RMPC. The y-axis shows how many cases (out of the 500) falls into each range for the two approaches. 
%	: we classified 500 samples with respect to the different fuel consumption improvement degree by DRL-based intermittent control and bang-bang control. We can see that in most cases, bang-bang control has a low degree of improvement. Comparing with bang-bang control, DRL-based intermittent control have a much higher probability to achieve more fuel saving.
}
	\label{fig:acc_example_3}
\end{figure}

%For the intermittent control, among 100 sampling steps, the average number of the steps that applying zero control input is 79.4, while the average number of states out of the strengthened safe set $X'$ is 8.6.

We also measured the computational savings from skipping the RMPC control computation. For our DRL-based opportunistic intermittent-control approach, the computation time for checking the satisfaction of strengthened safe set $X'$ and invoking the neural network to decide skipping choice $z$ is in average 0.02 second; while the average computation time for RMPC is 0.12 second. In our experiments, out of 100 steps, the average number of steps that skip the RMPC computation is 79.4. Thus, overall, there is around $60\%$ saving in computation time from our approach (i.e., $(0.12 \times 100 - (0.02 \times 100 + 0.12 \times (100 - 79.4)))/(0.12 \times 100)$).

%We also estimate the overhead of intermittent control. The average overhead -- the computation time of checking the satisfaction of strengthened safe set $X'$ and invoking neural network to obtain action $z$ -- at one time step of all 500 cases is about 0.02 seconds, while the average computation time of RMPC is 0.12 seconds. We can see that the overhead of our approach is one order of magnitude smaller than RMPC. Compared to the small amount of computation resource overhead, the fuel consumption, which is multiple order of magnitude larger than the computation resource consumption, is significantly reduced by our approach.

\subsection{Impact Analysis under Different Driving Scenarios}

We further conducted a series of experiments to evaluate our approach under different driving scenarios, particularly when the \emph{Front} vehicle exhibits different driving patterns. We are interested to see whether our approach can effectively learn those patterns and leverage them in achieving energy savings.

%With respect to the ACC scenarios and drive habits, the front car usually have different driving pattern, namely the way the front car velocity $v_f$ changes. In this section, we try to explore the performance how our approach reacts to the patterns.  

% The fuel consumption of \emph{Ego} will be impacted with different driving patterns, i.e., different perturbance patterns. Besides the \textbf{sinusoidal velocity pattern} used in previous experiment, we also study two more patterns:
% \begin{enumerate}
%     \item \textbf{Random acceleration pattern}: The acceleration of \emph{Front} is randomly chosen from $\Dot{v}_f(t)\in[-20, 20]$, while the corresponding velocity is bounded $v_f(t)\in[v_f^{min}, v_f^{max} ]$.
%     \item \textbf{Random speed pattern}: The acceleration of \emph{Front} is not restricted, while the velocity is randomly chosen from $v_f(t)\in[30, 50]$.
% \end{enumerate}

\smallskip
\noindent
\textbf{Impact of velocity range of \emph{Front} vehicle.}
%Our first set of experiments is designed to analyze the impact of the real range of the front cat velocity. 
We first analyze how different velocity range of the \emph{Front} vehicle may affect the performance of our approach. We conduct 5 experiments Ex.1 -- Ex.5, and the range of $v_f$ of these experiments are shown in Table~\ref{tab:set1}.
%In detail, we do 5 experiments Ex.1 - Ex.5, and the range of $v_f$ of these examples are shown in Table \ref{tab:set1}. 
We also restrict the \emph{Front} vehicle acceleration $v'_f$ with a bounded range $[-20,20]$. For each experiment, we test 500 cases by randomly picking feasible initial states within $X'$ and random front car acceleration $v'_f$ at each time step within $[-20,20]$. 

\begin{table}[htbp]
  \centering
  \caption{$V_f$ setting for Ex.1 -- Ex. 5}
    \begin{tabular}{cccccc}
     \hline
          & Ex. 1 & Ex. 2 & Ex. 3 & Ex. 4 & Ex. 5 \\
     \hline
    Range of $v_f$ & $[30,50]$ & $[32.5,47.5]$ & $[35,45]$ & $[38,42]$ & $[39,41]$ \\
     \hline
    \end{tabular}%
  \label{tab:set1}%
\end{table}%

The experimental results are shown in Figure~\ref{fig:jitter}. We can observe that when the range of $v_f$ becomes smaller, our DRL-based opportunistic intermittent-control approach achieves more fuel consumption savings over only using RMPC. This is because that a smaller range of $v_f$ is easier for DRL to learn and leverage.
%a trend from the Ex.1 - Ex.5 -- when the actual range of $v_f$ is becoming smaller, DRL-based intermittent control can achieve more fuel saving compared to RMPC. The reason is that the smaller the range of $v_f$, the more regularized the system is, which is more explicit for DRL to learn. Meanwhile, due to pessimistic estimation of the $v_f$ range ($[30,50]$), RMPC becomes more conservative when the actual range is smaller.  
%
\begin{figure}[tbp]
	\centering
	\includegraphics[width=0.38\textwidth]{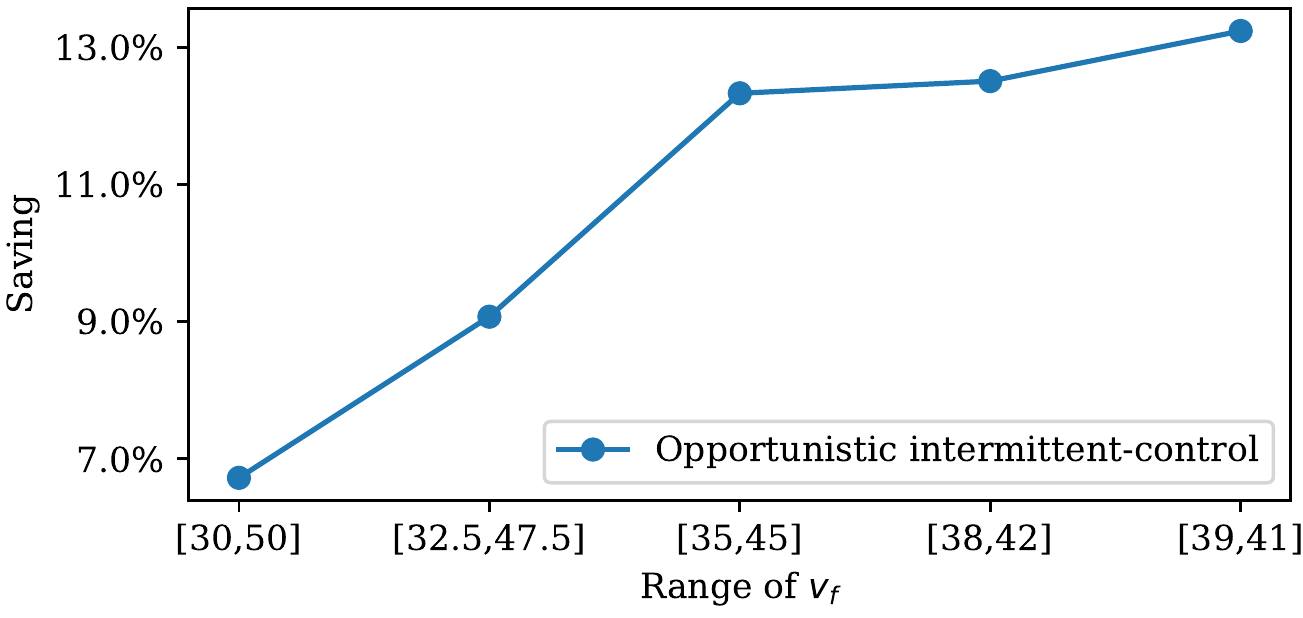}
	\caption{Fuel consumption savings by our DRL-based opportunistic intermittent-control over RMPC only under different range of $v_f$.}
	\label{fig:jitter}
\end{figure}

\medskip
\noindent
\textbf{Impact of velocity regularity of \emph{Front} vehicle.} 
We then conduct experiments to evaluate the impact of the regularity of the \emph{Front} vehicle velocity, i.e., how ``random'' the \emph{Front} vehicle changes its speed. We conduct experiments Ex.6 -- Ex.10, which share the same range of $v_f$ at $[30, 50]$ but with the following differences.
%Our second set of experiments is designed to explore the impact of the regularity of the front cat velocity. In detail, we do 5 experiments Ex.6 - Ex.10, which share the same range of $v_f$ but with the following difference.
\begin{itemize}
    \item   In Ex.6, $v_f$ changes completely random, i.e., a drastic change is allowed instantly;
    \item   Ex.7 shares the same setting with Ex.1, i.e., the velocity can only change continuously;
    \item   In Ex.8, the velocity changes based on Equation~\eqref{equ:sin_distr} but with a large random disturbance. Specifically, the amplitude $a_f=5$ and the range of $w$ is $[-5,5]$;
    \item  The setting of Ex.9 is similar to Ex.8, with more regularity, i.e., a larger amplitude $a_f=8$ and a smaller disturbance range $[-2,2]$;
    \item  The setting of Ex.10 is similar to Ex.8 and Ex.9, with even more regularity: i.e., $a_f=9$ and the disturbance range is $[-1,1]$.  
    %Ex.10 likewise considers the sinusoidal velocity pattern, but with more certainty: $a_f=9$ and the disturbance range is $[-1,1]$.
\end{itemize}

Intuitively, from Ex.6 to Ex.10, the \emph{Front} vehicle velocity exhibits more regularity. The experiments results are shown in Figure~\ref{fig:chaos}. From Ex.7 to Ex.10, we can see that better regularity leads to easier learning of our DRL-based approach and more fuel consumption savings. However, Ex.6 is an exception, where $v_f$ is purely random but our approach still achieves significant saving. We speculate that this phenomenon is due to the low performance of the RMPC itself, since there would be a mismatch between the real state and what RMPC predicts.
%Overall, we can see that when $v_f$ changes more regularly, DRL-based intermittent control can achieve more fuel saving, which is consistent with the result in the first experiment set. However, Ex. 6 is a singular point, where $v_f$ is purely random without any learnable knowledge but our approach still achieves improvements. We speculate that this phenomenon is due to the low performance of RMPC itself, since there would be a mismatch between the real state and what RMPC predicts. It still needs further study to figure out the concrete reason behind. 
%
\begin{figure}[tbp]
	\centering
	\includegraphics[width=0.38\textwidth]{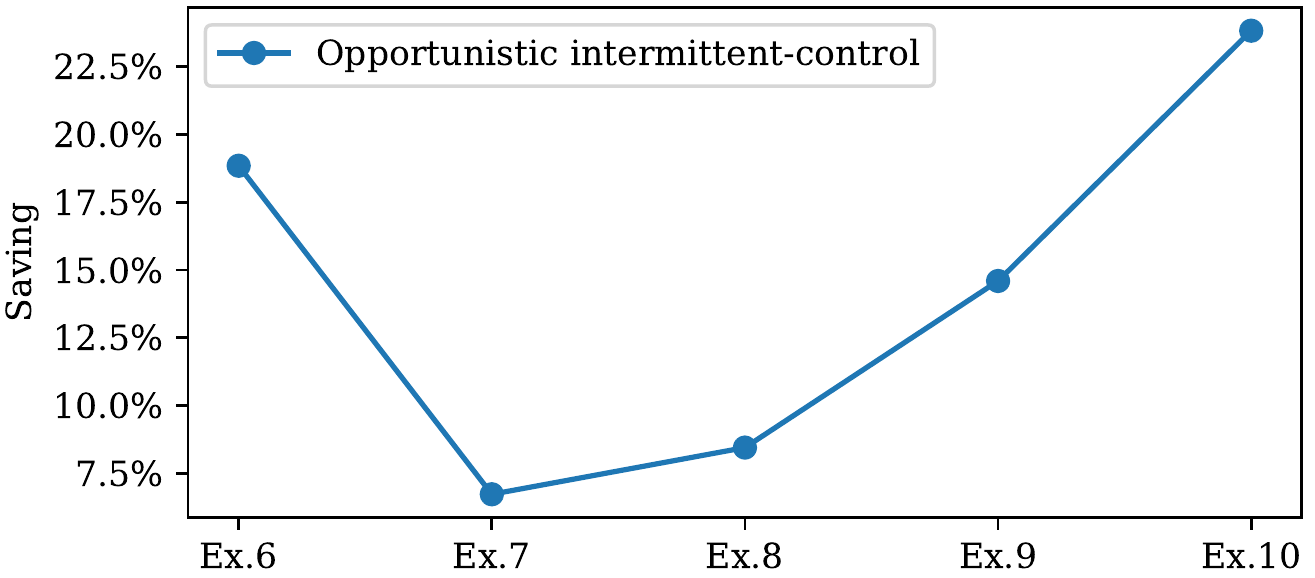}
	\caption{Fuel consumption savings by our DRL-based opportunistic intermittent-control under different regularity degree of $v_f$.}
	\label{fig:chaos}
\end{figure}

%For the sinusoidal velocity pattern, we also did several experiments with different $v_e$, $a_f$ and disturbance range of $w$ of~\eqref{equ:sin_distr}. And for the random acceleration pattern, we did experiment for different velocity bound $[v_f^{min}, v_f^{max} ]$. For each specific pattern, we generated 500 cases and the average fuel saving of the intermittent control is shown in Table~\ref{tab:diff_patterns}.

\section{Conclusion} 
\label{sec:conclusion}

It is often overly conservative to design safety-critical systems with only the worst-case situations in mind.
%such that the system can be safe even in the worst case. 
In this paper, we propose a novel online intermittent-control framework to opportunistically skip the computation and actuation of an underlying safe controller via learning the knowledge of the specific operation context and environment. The framework utilizes both a model-based approach (mixed integer programming formulation) and a learning-based approach (deep reinforcement learning) for different situations to intelligently make the skipping decisions. It also guarantees system safety by formally computing a strengthened safe set based on the notion of robust control invariant and backward reachable set of the underlying controller. 
%effectively leverage the knowledge of operation contexts. To guarantee the safety, we compute a strengthened safe set based on the notion of robust control invariant and backward reachable set of the underlying controller, which forms a decision condition of skipping controls. Then a monitor is deployed to check if the state is within the strengthened safe set and decide whether the control can be safely skipped. 
Experiments on an adaptive cruise control system demonstrate the effectiveness of our approach in achieving significant energy and computation savings. Future work includes addressing more complex control systems beyond LTI.

\bibliographystyle{IEEEtran}
\bibliography{IEEEabrv,./references_qi,./bib_zwang,./wenchao,./hyoseung,./Shichao,./chao}
% The following will be deleted when the paper is submitted.
%\appendix

%\input{sections/appendix.tex}
\end{document}